\documentclass[letterpaper, 10pt, conference]{ieeeconf}
\IEEEoverridecommandlockouts
\usepackage[letterpaper,top=0.75in,left=0.75in,right=0.75in,bottom=0.75in]{geometry}

\usepackage{blindtext}

\usepackage{caption}
\usepackage{amsmath, amsfonts, amssymb, mathtools, url, graphicx, amsthm, balance}
\usepackage{algorithm}
\usepackage[noend]{algpseudocode}
\usepackage{mathtools}
\usepackage{cleveref}
\usepackage{relsize}
\DeclareMathOperator*{\argmin}{arg\,inf}

\newcommand{\Rbb}{\mathbb{R}}
\newcommand{\Pbb}{\mathbb{P}}

\newtheorem{definition}{Definition}
\newtheorem{assumption}{Assumption}
\newtheorem{lemma}{Lemma}
\newtheorem{proposition}{Proposition}
\newtheorem{theorem}{Theorem}
\newtheorem{example}{Example}
\newtheorem{remark}{Remark}
\begin{document}
%
\title{Minimal Realization Problems for Jump Linear Systems}

 \author{T. Sarkar, M. Roozbehani, M. A. Dahleh
 	\thanks{T. Sarkar, M. Roozbehani, M. A. Dahleh are with the EECS
 		Department, Massachusetts Institute of Technology, Cambridge, MA 02139 USA (e-mail: \texttt{tsarkar, mardavij,dahleh}@mit.edu).}
 }

\markboth{Journal of \LaTeX\ Class Files,~Vol.~6, No.~1, March~2018}%
{Shell \MakeLowercase{\textit{et al.}}: Bare Demo of IEEEtran.cls for Journals}

\maketitle

 \begin{abstract}
This paper addresses two fundamental problems in the context of jump linear systems (JLS). The first problem is concerned with characterizing the minimal state space dimension solely from input--output pairs and without any knowledge of the number of mode switches. The second problem is concerned with characterizing the number of discrete modes of the JLS. For the first problem, we develop a linear system theory based approach  and construct an appropriate Hankel--like matrix. The rank of this matrix gives us the state space dimension. For the second problem we show that minimal number of modes corresponds to the minimal rank of a positive semi--definite matrix obtained via a non--convex formulation.
 \end{abstract}
\section{Introduction}
Jump linear systems (JLS) are abstractions of hybrid systems, that combine continuous and discrete dynamics. They extend LTI systems in the sense that state space updates in a linear fashion contingent on the underlying discrete parameter. These are also known as multi--modal systems, due to multiple transitions or modes of operation controlled by the discrete parameter. In this work we assume that these transitions are independent of the past, \textit{i.e.}, the discrete parameter is an independent and identically distributed random variable. JLSs have been used in economics~\cite{do1999receding} where each mode corresponds to a business cycle and the state is economic growth; in statistics~\cite{fox2010bayesian} to model multiple related time series, where each series is generated from a single mode;  and large infrastructure systems such as air traffic networks~\cite{hamsa2017}.

These hybrid systems have been widely studied in control theory (See~\cite{costa2006discrete} for a detailed review on stability and optimal control). In this work we are concerned with the realization theory of JLS. Realization theory of LTI systems has been well studied by Kalman et. al. in~\cite{ho1966effective},~\cite{tether1970construction}. This was achieved by constructing a Hankel matrix mapping the past inputs to future outputs and minimal realization was characterized by the rank of this Hankel matrix. These results lead to a cascade of literature on linear system identification (See~\cite{ljung1998system} for a full review). The second set of results in LTI relates to model reduction. The objective is to find a simplified model for a complex system. Balanced truncation and optimal Hankel reduction are two related reduction algorithms accompanied by provable misspecification bounds(See~\cite{glover1984all}).

Hybrid systems have been widely studied in control theory (See~\cite{costa2006discrete} for a detailed review on stability and optimal control). The realization problem is understood to a large extent. Notions of controllability and observability are corner stones of realization in LTI systems theory. However, these do not extend to jump linear systems in an obvious fashion. For JLS authors in~\cite{ji_controllability},~\cite{ge2001reachability} define controllability and observability in a weak sense. Specifically, they present verifiable conditions under which a JLS is (weakly) controllable or observable. However, the algorithmic aspects of system identification are not discussed there. A more complete approach towards the realization theory is studied in~\cite{petreczky2007realization},\cite{petreczky2007metrics},\cite{petreczky2010realization},\cite{petreczky2018}. In ~\cite{petreczky2007realization},\cite{petreczky2007metrics},\cite{petreczky2018} there is no external input and discrete modes are observed. They discuss necessary and sufficient conditions for the existence of a realization, provide a characterization of minimality in jump linear systems that can be checked algorithmically. In~\cite{petreczky2010realization}, authors consider the case where there is no noise but there is a control input. The necessary and sufficient conditions for realizations are provided in terms of the rank of observability and reachability span. The case of unknown discrete modes with deterministic switching is considered in~\cite{westra2011identification}. The trade--off between the dimension continuous state variable and number of discrete modes is discussed there. 

In our work we assume the discrete modes are unknown and unobserved. We present results for the realization of a JLS when only the continuous input and output are observed. We believe that the unobservability of discrete modes makes more sense in the case of, for example, macroeconomic policy design in economics. There each business cycle corresponds to a discrete state (or mode) which is unknown apriori. The economic growth is then a function of the underlying mode. In this work, we first show how to find the minimal state space dimension. This involves finding the rank of a certain Hankel--like matrix mapping past inputs to future outputs. With the knowledge of state space dimension, we obtain the number of modes in the JLS that involves solving a non--convex problem. The number of modes corresponds to finding the rank of a positive semi--definite matrix that is obtained by transforming the output observation matrix in a certain way. We also discuss how finding discrete modes of a JLS is related to intrinsically hard problems such as finding the minimum number of mixture components in a Gaussian Mixture Model (GMM). Our focus here is not JLS identification rather we would like to characterize the minimal values of $(n, s)$, \textit{i.e.}, state space dimension and number of discrete parameters respectively, required to completely span the output space corresponding to any arbitrary length input sequence. 

The paper is arranged as follows -- in Section~\ref{model} we describe in detail our model and elucidate the problems in JLS identification. In Section~\ref{algorithm} we describe our algorithms to find the minimal realization and number of modes; in Section~\ref{results} we state and prove our main results including the correctness of our algorithms. Finally, we conclude in Section~\ref{conclusions}.
\section{Mathematical Notation}
Any vector $v \in \mathbb{R}^{p}$ is a $p \times 1$--column vector. The positive semi--definite cone is denoted by $\mathcal{S}_n^{+}$. For any $m \times n$ matrix, $A$, we denote by $A^{'}$ its transpose. The Frobenius norm of a matrix $A$ is given by $||A||_F^2 = \sum_{i, j}A_{ij}^2$. $X^{\dagger}$ is Moore--Penrose pseudo inverse of $X$. Define $[s] = \{1, 2, \ldots, s\}$, then $\sigma \in [s]^T$ denotes a permutation of length $T$ on the set $[s]$. For a matrix $A$, the operator $\text{vec}(A)$ vectorizes the matrix by stacking the columns one top of the other. We denote by $A \otimes B$ as the Kronecker product of $A$ and $B$.
\section{System Model}
\label{model}
Consider the following switched linear system (also called a jump linear system, or JLS for short)
\begin{align}
    \label{jls}
    x_{k+1} &= A[\theta(k)]_{n \times n} x_k + B_{n \times p} u_k \nonumber\\
    y_k &= C_{m \times n} x_k
\end{align}
Here $\theta(k) \in \{1, 2, \ldots, s\}$, where $i \in \{1, 2, \ldots, s\}$ is a mode (or switch). Further $\{\theta(k)\}_{k=1}^{\infty}$ are i.i.d with $\Pbb(\theta(k) = j) = p_j$. As is obvious from Eq.~\eqref{jls} $\{A[l]\}_{l=1}^s, B, C$ are $n \times n, n \times p, m \times n$ matrices respectively. The input $u_k$ is $p \times 1$ and the output $y_k$ is $m \times 1$. In general we only observe $\{(u_k, y_k)\}_{k=1}^{\infty}$ and the system parameters $\{s, \{A[l]\}_{l=1}^s, C, B\}$ are unknown. Consequently, the dimension of $A[l]$, \textit{i.e.}, $n$ is unknown. Throughout this paper we will also assume that the initial conditions $x_0 = 0$. Under mean--square stability of the JLS, this can be achieved by letting the JLS run to $t = \infty$ without any inputs. Although this is not a necessary assumption, it eases the analysis.

The goal of this paper is to determine the ``minimal'' $(s, n)$ pair that satisfies the observed output. The definition of minimality will be made precise later. A JLS is an inherently probabilistic system, due to transitions of $\theta(k)$; we introduce two parallel notions of such dynamical systems that will assist in determining identifiability of JLS. The first notion is the traditional control theoretic interpretation of a dynamical system, where we observe the input--output pair $\{(u_t, y_t)\}_{t=1}^{\infty}$. We borrow our second notion from statistical learning theory, where we run $N$ independent copies of the dynamical system, each for time $T$. Traditionally, $N, T$ are known as the \textbf{sample} and \textbf{information} complexity respectively. Informally, a large $N$ helps us achieve close approximations to various expected system parameters (such as expected system energy etc.) and a large $T$ helps us achieve the complete span of unobserved state variable $x(\cdot)$, $\grave{\text{a}}$ la Hankel approach. Together these parameters give us an idea of minimal $(s, n)$ pair. The advantage behind the second notion is that it helps us visualize and subsequently bound the information and sample complexity separately. Fig.~\ref{two_approach} encapsulates the ideas of the two notions. 
\begin{figure}[h]
\centering
    \includegraphics[width=0.7\columnwidth]{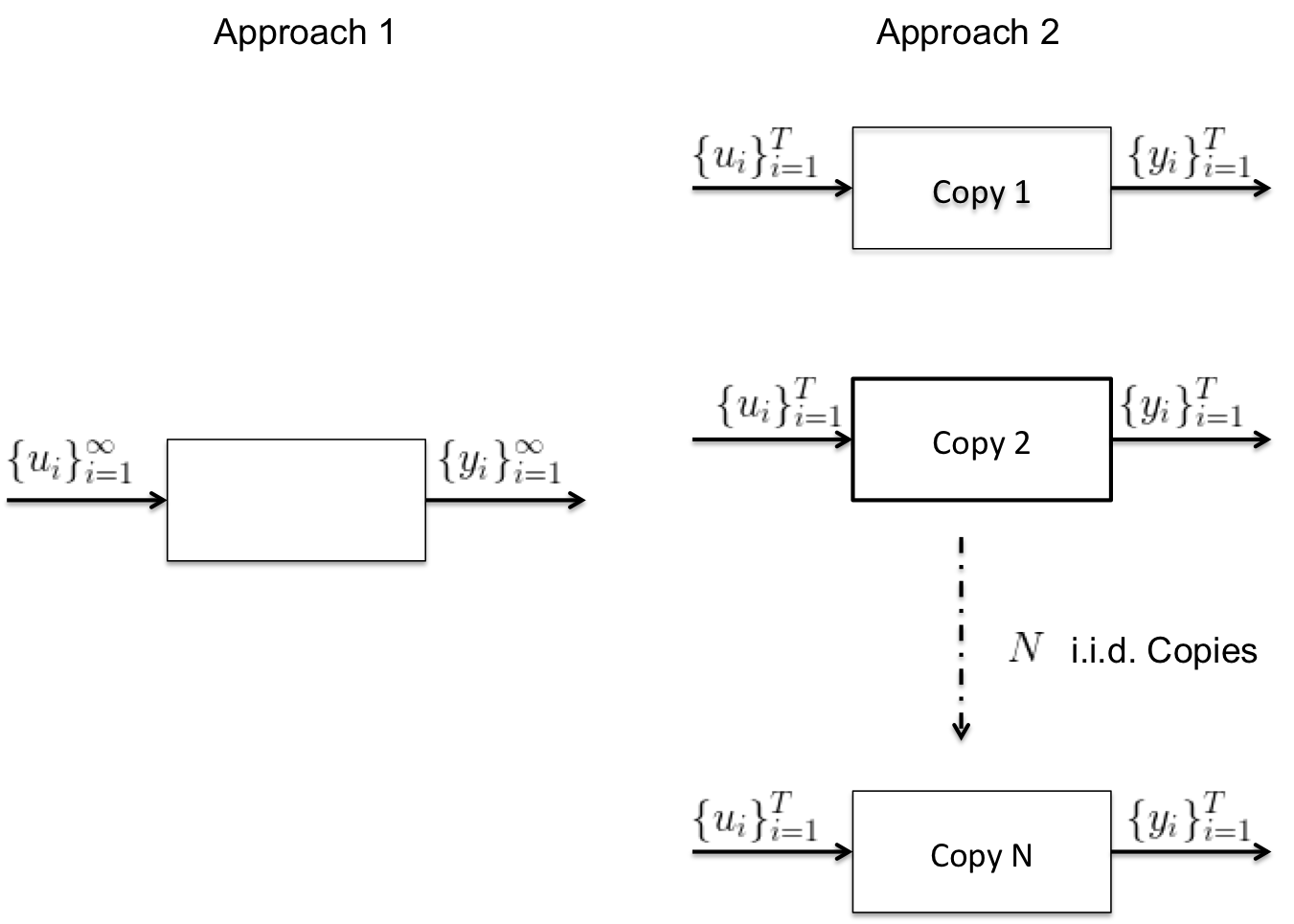}
    \caption{Different approaches to view JLS}
    \label{two_approach}
\end{figure}

Each copy $m$ should be thought of as a specific occurence of switches in the JLS. For example, consider a JLS where $s=2$, and if outputs are collected up to time $T = 2$ then for $x_0 = 0$
\begin{align*}
(A[\theta(1)], A[\theta(2)]) &\in \{(A_1, A_1), (A_1, A_2), (A_2, A_1), (A_2, A_2)\} \\
(y_1, y_2) &= (C A[\theta(1)] B, CA[\theta(2)]B)
\end{align*}
Copy $m$ corresponds to any specific occurence of $(A[\theta(1)], A[\theta(2)])$ and can be uniquely characterized as $i = (\theta(1), \theta(2)) \in [s]^2$. In general, $m = \sigma \in [s]^T$. It is the hope that $N$ is large enough that all possible switch sequence occurences happen sufficiently often. This is akin to running the JLS with $N$ restarts ($x_0= 0$). Due to this stochastic nature of a JLS, traditional notions of controllability and observability in LTI systems fail to apply. We follow the approach of~\cite{ji_controllability} and define appropriate concepts of stochastic controllability and observability.
\begin{definition}
\label{weak_controllability}
For the JLS in~\eqref{jls}, every initial state $x_0, \theta(0)$ and any $x \in \mathbb{R}^n$, if there exists an admissible control law $\{u_1, u_2 , \ldots\}$ such that the random time $T_{cw}$
\begin{equation*}
    T_{cw}(x) = \inf_{u \in U} \{k > 0: \mathbb{P}(x(k, x_0, \theta(0), u) = x) > 0\}
\end{equation*}
is finite almost surely, then the system is weakly controllable.
\end{definition}
Define $y_{\{x_{\text{init}} = x_1\}}(k)$ as the output at time $k$ when initial state was $x_1$ at time $k_0$ (or any initial time).
\begin{definition}
\label{weak_observability}
For the JLS in~\eqref{jls}, any initial state $\theta(0)$ and two initial $x$-states $x_1, x_2$, let $T_{ow}$, be the minimum time such that if the outputs $y_{\{x_{\text{init}} = x_1\}}(k) = y_{\{x_{\text{init}} = x_2\}}(k)$ for all $k$ between $k_0$ and $T_{ow} + k_0$ and the inputs u, are known in this time interval, then $\mathbb{P} (x_1 = x_2) > 0$. We say that this system is weakly observable if $T_{ow}$ is finite almost surely.
\end{definition}
We now provide some intuition for the two definitions. First, we show that weak observability and weak controllability correspond to observability and controllability in LTI systems respectively.
\begin{example}
\label{LTI}
For the case of LTI systems $s = 1$. Therefore, for any initial state $x_0$ ($\theta(k) = 1$ identically for all $k$), and any $x$, there exists an input sequence such that $T_{cw}(x) \leq n$. If the system was not controllable then there exists $x^{'}$ that is never reached by the state. As a result $T_{cw}(x^{'}) = \infty$. A dual argument can be made for observability.
\end{example}
\begin{example}
\label{s2}
Consider the case where $s=2$ with $B = [1, 0, 0]^T$, $C = [1, 0, 0]$, $A_1 = I_{3 \times 3}$ and 
\begin{align*}
    A_2 = \begin{bmatrix}
    0  & 0 & 1 \\
    1 & 0 & 0 \\
    0 & 1 & 0
    \end{bmatrix}
\end{align*}
It is not hard to observe that $\mathcal{B} = [B, A_2B, A_2^2 B]$ has column rank $3$ while 
\[
\mathcal{C} = \begin{bmatrix}
C \\
CA_2 \\
CA_2^2
\end{bmatrix}
\]
has row rank $3$. Assume that $p(\theta(k) = l) = 1/2$ then we for any initial $x_0, \theta(0)$ and final $x$, we have that with probability at least $1/4$, we can reach $x$ in finite time. Consider an input sequence $\{u_3, u_2, u_1, u_0\}$ such that $Bu_3 + A_2Bu_2 + A_2^2 Bu_1 + A_2^2 A_{\theta(0)} (u_0 + x_0) = x$. Such an input exists because $\mathcal{B}$ has rank $3$. The probability of this happening is $\mathbb{P}(\theta(2)=2, \theta(1) = 2) = 1/4$ and $T_{cw} \leq 4$ for all $x$. Thus the system is weakly controllable. 

Although for this JLS we know that there is a certain sequence of switches that allow the state space to span its full dimension, we do not know when (or if at all) it occurs by just observing a sequence of outputs. Consider another occurence where $\theta(1) = \theta(2) = 1$ which again can happen with probability $1/4$ then 
$\mathcal{B}^{''} = [B, A_1B, A_1^2 B]$ has column rank $1$ while 
\[
\mathcal{C}^{''} = \begin{bmatrix}
C \\
CA_1 \\
CA_1^2
\end{bmatrix}
\]
has row rank $1$. Therefore, under no knowledge of the hidden state it is not possible to determine the true state space dimension the same way as one would do for an LTI system.
\end{example}
We now introduce alternate but equivalent definitions of (weak) controllability and observability that will be useful.
\begin{definition}[Controllability Rank]
\label{reachability}
The controllability rank of a JLS defined in~\eqref{jls} is $r_B$, \textit{i.e.}, 
\begin{align*}
 \text{Column Space}(B) = \bigcup_{\sigma \in [s]^{\infty}}\text{Column Space}(&[B, A_{\sigma_1}B, \\
 &A_{\sigma_2} A_{\sigma_1}B, \ldots])   
\end{align*}
Then the controllability rank is given by $r_B = \text{Rank}(\text{Column Space}(B))$.
\end{definition}

\begin{definition}[Observability Rank]
\label{observability}
The observability rank of a JLS defined in~\eqref{jls} is $r_C$, \textit{i.e.}, 
\begin{align*}
\text{Row Space}(C) &= \bigcup_{\sigma \in [s]^{\infty}}\text{Row Space}\Bigg(\begin{bmatrix} C \\
 CA_{\sigma_1} \\
 CA_{\sigma_2}A_{\sigma_1}\\
 \vdots \end{bmatrix}\Bigg)   
\end{align*}
Then the observability rank is given by $r_C = \text{Rank}(\text{Row Space}(C))$.
\end{definition}
\begin{proposition}
\label{definition_equiv}
JLS is weakly observable if and only if observability rank is $n$. Similarly, JLS is weakly controllable if and only if the controllability rank is $n$. 
\end{proposition}
\begin{proof}
The proofs of this can be found as Theorems 2 and 5 in~\cite{ji_controllability}. This follows by assuming that $\mathbb{P}(\theta(k) = i) = p_{ii} = p_{ji}$ for all $j, i$.
\end{proof}
Definition~\ref{weak_controllability},~\ref{weak_observability} are based on Approach $1$ in Fig.~\ref{two_approach}. In fact Theorem 2, 5 in~\cite{ji_controllability} suggest that if a JLS is weakly controllable (or observable) then there exists a sequence of switches for which the controllability matrix is full rank \textit{i.e.}, there exists $\sigma^{L} \in [s]^L$ such that $B_L$ is rank $n$. 
\begin{equation*}
    B_L = \begin{bmatrix}
    B, A_{\sigma^L_{L-1}}B, \ldots, A_{\sigma^L_{L-1}} \hdots A_{\sigma^L_{1}}B 
    \end{bmatrix}
\end{equation*}
However, a major limitation of this result is that $L$ is unknown (even an upper bound). Further, for a JLS, $B_L$ can never be explicitly observed. In this work we will show that, with the help of Definitions~\ref{reachability},~\ref{observability} and the second approach shown in Fig.~\ref{two_approach}, there exists a value of $T$, also referred to as the information complexity, such that their maximum ranks are achieved. We next define $\mathcal{H}_{\mu^{T}, \sigma^{T}}$, a Hankel--like structure, recurring in our analysis.
\begin{definition}
	\label{hankel}
	For $\sigma^T, \mu^T \in [s]^T$, we define $C_{\mu^T}, B_{\sigma^T}$ as 
	\begin{align*}
	C_{\mu^T} &=\begin{bmatrix} C \\
	CA_{\mu^T_1} \\
	CA_{\mu^T_2}A_{\mu^T_1}\\
	\vdots \\
	CA_{\mu^T_T}A_{\mu^T_{T-1}}\ldots A_{\mu^T_1}
	\end{bmatrix} \\
	B_{\sigma^T} &= [B, A_{\sigma^T_1}B, A_{\sigma^T_2} A_{\sigma^T_1}B, \dots, A_{\sigma^T_T} \ldots A_{\sigma^T_2}A_{\sigma^T_1}B] 	
	\end{align*}
	Then $\mathcal{H}_{\mu^T, \sigma^T} = C_{\mu^T} B_{\sigma^T}$.
\end{definition}
\section{Algorithmic Setup}
\label{algorithm}
Recall the setup in Eq.~\eqref{jls} where only input--output pairs, \textit{i.e.}, $\{(u_t, y_t)\}_{t=1}^{\infty}$ are observed with $x_0 = 0$. We justified this by assuming a mean--square stable JLS. Equivalent conditions of stability can be found in Theorem 2.1~\cite{kotsalis2008balanced}.
\begin{assumption}
	\label{stability}
The JLS in Eq.~\eqref{jls} is mean--square stable.
\end{assumption}
For the purpose of our algorithm we will also assume that the JLS is both weakly controllable and weakly observable. If that was not the case, then our algorithm would find the maximal common subspace of observability and controllability. However, we do not go into details of such a case in this work and simply assume weak controllability and observability for simplicity (although the state space dimension is unknown). Further assume that we have $N$ i.i.d. samples of the JLS, \textit{i.e.}, we have $N$ identical copies of the JLS where the switches occurring in each copy are independent of another copy. We will assume that $N$ is sufficiently large so that the empirical average of any function of outputs is almost equal to its expectation.
\begin{assumption}
\label{minimal_rank}
JLS in Eq.~\eqref{jls} is weakly controllable and weakly observable. 
\end{assumption}
\begin{definition}
\label{minimal_rank_def}
Under Assumption~\ref{minimal_rank}, the minimal rank, $n$, of the JLS in Eq.~\eqref{jls} is the controllability (or observability) rank. The minimum number of discrete modes (corresponding to $n$), $s_n$, that completely describes the output space in response to inputs of arbitrary length is the minimal number of discrete modes. Together the pair $(n, s_n)$ is called the minimal realization of a JLS. 
\end{definition}
If the JLS is such that the state space dimension and discrete modes equal the minimal realization of the JLS, then it is a minimal JLS.
\begin{example}
\label{minimal_s}
The number of discrete modes needed to represent the output space is not unique, as one could always overparametrize the JLS. For a parsimonious approach we introduced the notion of minimality. Consider $C = [1, 1, 1], B = [1, 1, 1]^{'}, p_i = 1/4$ with 
\begin{align*}
A_1 = \begin{bmatrix}0, 0, 0 \\ 1, 1, 0 \\-1, 0, 0\end{bmatrix} &, A_2 = \begin{bmatrix}1, 0, -1 \\ -1, 0, 0 \\1, 0, 1\end{bmatrix} \\
A_3 = \begin{bmatrix}0, 0, 1 \\ 0, 0, 0 \\0, 0, 0\end{bmatrix} &, A_4 = \begin{bmatrix}1, 0, 0 \\ 0, 1, 0 \\0, 0, 1\end{bmatrix}
\end{align*}
Then it is easy to check that the JLS is controllable and observable. But the minimal $s \leq 3$, the reason is that $A_4 = A_1 + A_2 + A_3$, \textit{i.e.}, it is a linear combination of the other modes. As a result, it does not matter if we exclude $A_4$ (in an output sense), \textit{i.e.}, any output generated by $A_4$ can be generated by a linear combination of $A_i$. The exact nature of minimality is described in Proposition~\ref{switches}.
\end{example}
\begin{assumption}
\label{sample_complexity}
Let $y^{m}_t$ denote the output at time $t$ from copy $m$. Then we assume that $N$ is large enough so that for any bounded differentiable function $f(\cdot)$ we have that 
\begin{equation*}
    \frac{1}{N}\sum_{i=1}^N f(y^{m}_1, \ldots, y^{m}_T) = \mathbb{E}[f(Y_1, \ldots, Y_T)]
\end{equation*}
where $Y_i$ is the random variable indicating output at time $i$ for a fixed T.
\end{assumption}
\begin{remark}
\label{sample_remark}
The sufficient sample complexity depends on the value of $T$ and the probabilities of switches. For example let the sequence of switches that makes the JLS controllable be given by some $\sigma^T$ with $B_{\sigma^T} = [B, A_{\sigma^T_{T-1}}B, \ldots, A_{\sigma^T_{T-1}}\ldots A_{\sigma^T_1}B]$, and all other switch sequences give $0$ output. Further assume that $T$ is minimal in the sense that there is no switch sequence for $T^{'} < T$ that makes the JLS controllable. Then to conclude weak controllability in this pathological system one needs to observe this switch sequence at least once. The expected time of observing this is given by $T_O = (p_{\sigma^T_1}p_{\sigma^T_2} \ldots p_{\sigma^T_{T-1}})^{-1}$ where $p_{\sigma^T_j} = \mathbb{P}(\theta(j) = \sigma^T_j)$. In the worst case, sample complexity $N = \Omega(T_O)$. Observe that depending on the probabilities in might be exponential in $T$. As a result it is of interest to find the relationship between $N$ and $T$ (the information complexity). 
\end{remark}
\begin{assumption}[Orthogonal Differentiability]
	\label{independent_subspace}
	There exists $T^{*}$ such that 
	\begin{align*}
	\mathbb{E}_{\mu \in [s]^{T^{*}}}[C_{\mu} P_1 C^{'}_{\mu}] &\neq \mathbb{E}_{\mu \in [s]^{T^{*}}}[C_{\mu}P_2 C^{'}_{\mu}] \\
	\mathbb{E}_{\mu \in [s]^{T^{*}}}[B^{'}_{\mu} P_1 B_{\mu}] &\neq  \mathbb{E}_{\mu \in [s]^{T^{*}}}[B^{'}_{\mu} P_2 B_{\mu}]
	\end{align*}
	where $P_l \succeq 0$ and $\langle P_1, P_2 \rangle = 0$.
\end{assumption}	
\begin{remark}
	\label{validity_assumption}
	We would like to know how restrictive Assumption~\ref{independent_subspace} is in the case of JLS. Informally, it means that for orthogonal vectors the subspace generated by observability (and controllability) matrix should not be identical. This seems natural because since we do not observe hidden state variables (or internal state switches) if orthogonal state vectors generated the same output subspace we would not be able to differentiate between them. 
\end{remark}
The basis of our approach will be divided in two parts: excitation and observation. These would roughly correspond to controllability and observability of the JLS, as we discuss below. In the excitation stage, we excite different copies of the dynamical system for $t \in [0, T-1]$. Then we collect the output for $t \in [T, 2T]$ with $u_k = 0$ during that time. As discussed in the previous section, a copy $m$ simply corresponds to a sequence of switches, \textit{i.e.}, $m = (\mu^T \in [s]^{T}, \sigma^T \in [s]^T)$. Then one can show that 
\begin{align}
\label{hankel_eq}
\underbrace{\begin{bmatrix} y_T \\y_{T+1} \\ \vdots \\ y_{2T-1}\end{bmatrix}}_{y_{+}} &= \mathcal{H}_{\mu^T, \sigma^T} \underbrace{\begin{bmatrix} u_0 \\u_1 \\ \vdots \\ u_{T-1}\end{bmatrix}}_{u_{-}}
\end{align}
The relation in Eq.~\eqref{hankel_eq} is the Hankel matrix for copy $m$ and will be useful in the analysis of our algorithm for JLS realization.
\begin{definition}
	\label{output_jls}
	For every switch sequence occurence (or copy) $m$ of the JLS in Eq.~\eqref{jls}, with input $\{u_k\}_{k=0}^{\infty}$, we define
	 \begin{align*}{}
	 Y_m(t, \tau, \{u_k\}_{k=0}^{\infty}) = \begin{bmatrix}
	 y^{(m)}_t \\
	  y^{(m)}_{t+1} \\
	 \vdots \\
	 y^{(m)}_{\tau}
	 \end{bmatrix}
	 \end{align*}
	 as the outputs from time $t$ to $\tau$. 
\end{definition}
In our set up all our inputs have a finite support in the excitation stage $[0, T-1]$. Therefore, $Y_m(t, \tau, \{u_k\}_{k=0}^{\infty}) =  Y_m(t, \tau, \{u_k\}_{k=0}^{T-1})$. Indeed due to causality the outputs from $t$ to $\tau$ only depend on inputs from $0$ to $\tau$. Then for every $v \in \mathbb{R}^{pT}$ we have a unique bijection onto the space of inputs $\{u_k\}_{k=0}^{\infty}$ with $u_k = 0$ for all $k \geq T$. We define that map $\mathcal{B}$ as follows  
\begin{definition}
	\label{linearize}
For all $v \in \mathbb{R}^{pT}$ 
\begin{align*}
v = \begin{bmatrix}
u_0 \\
u_1 \\
\vdots \\
u_{T-1}
\end{bmatrix}
\end{align*}
where $u_j \in \mathbb{R}^{p}$, we define $\Pi$ as $\Pi(v) = \{u_j\}_{j=0}^{\infty}$ where $u_j = 0$ for all $j \geq T$.
\end{definition}
An important observation in our algorithm is that given $v$ that completely span $\mathbb{R}^{pT}$ we can obtain the minimal rank of a JLS. To that end we define an input basis set.
\begin{definition}
	\label{basis_set}
The basis set of inputs is $\mathcal{U} = \{v_1, v_2, \ldots, v_{pT}\}$ where $v_i \in \mathbb{R}^{pT}$ and $\text{dim}(\mathcal{U}) = pT$.
\end{definition}
We extend $\mathcal{U}$ to $\bar{\mathcal{U}}$ as 
\begin{equation}
\label{extension}
\bar{\mathcal{U}} = \mathcal{U} \cup_{i \neq j} \{v_i + v_j, v_i - v_j\}
\end{equation}
\begin{algorithm}[h]
	\caption{Learning Jump Linear Systems}
	\label{alg:learn_jls}
	\textbf{Input} $m= \text{Output\_Dimension(JLS)}$ \\
	$p = \text{Input\_Dimension(JLS)}$ \\
	$T=0.5 \times \text{Time to run an instance of the JLS}$ \\
	$N = \text{Sample Complexity}$\\
	\textbf{Output} Observation Matrices: $(Y_O, Y^{+}_O)$
	\begin{algorithmic}[1]
		\State Create basis set of inputs $\bar{\mathcal{U}} = \{\bar{v}_j\}_{j=1}^d$ as Eq.~\eqref{extension}
		\For {$j=1$ to $d$}
		\State $\Pi(\bar{v}_j)  = \{u_k^{(j)}\}_{k=0}^{\infty}$ 
		\For{$m=1$ to $N$}
		\State $\hat{Y}_{mj} = Y_m(T, 2T-1, \{u^{(j)}_k\}_{k=1}^{\infty})$ 
		\State $\hat{Y}^{+}_{mj} = Y_m(T+1, 2T, \{u^{(j)}_k\}_{k=1}^{\infty})$
		\EndFor
		\State  $\hat{Y}_j =\dfrac{1}{N}\mathlarger{\sum}_{m=1}^N \text{vec}(\hat{Y}_{mj} \hat{Y}^{'}_{mj})$
		\State  $\hat{Y}^{+}_j =\dfrac{1}{N}\mathlarger{\sum}_{i=1}^N \text{vec}(\hat{Y}^{+}_{mj} (\hat{Y}^{+})^{'}_{mj})$
		\EndFor
		\State $Y_O = [\hat{Y}_1, \hat{Y}_2, \ldots, \hat{Y}_{d}]$
		\Statex $Y_O^{+} = [\hat{Y}^{+}_1, \hat{Y}^{+}_2, \ldots, \hat{Y}^{+}_{d}]$
		\State \Return $(Y_O, Y^{+}_O)$
	\end{algorithmic}
\end{algorithm}
The output of Algorithm~\ref{alg:learn_jls} gives us two observation matrices $(Y_O, Y^{+}_O)$. We will see that $Y_O$ is useful in finding the minimal rank of the JLS and $Y^{+}_{O}$ will be used to find the number of different $A_l$s (or switches). To find the number of switches we need a swap transformation that we define below.
\begin{proposition}
\label{switches}
Let $S = \sum_{i=1}^s p_i(A_i \otimes A_i)$, then there exists a known linear bijective transformation such that $\mathcal{L}(S) = \sum_{i=1}^s p_i \text{vec}(A_i) \text{vec}(A_i)^T$. We refer to $\mathcal{L}(\cdot)$ as the swap transform.
\end{proposition}
We also define a non--convex problem that will be at the center of our minimal realization. 
\begin{align}
\mathcal{PF}(Y_O, Y_O^{+}, P, n) &= \text{rank}(P) \label{rank_minim}\\
\text{s.t. } P &\in \mathcal{S}_{n^2}^{+} \nonumber\\
Y_O = UV^{'} &\text{ s.t. } U \in \Rbb^{mT \times n^2}, V \in \Rbb^{pT \times n^2} \nonumber \\
U^{\dagger} Y_O^{+} V^{\dagger} &= \mathcal{L}^{-1}(P) \nonumber
\end{align}
Then the algorithm for learning JLS switches is given in Algorithm~\ref{alg:switches}. The output of this algorithm is a pair of matrices which we show can be used to find the parameters of the JLS.
\begin{algorithm}[h]
	\caption{Learning JLS Switches}
	\label{alg:switches}
	\textbf{Input} Observation matrices $Y_O, Y^{+}_O$ \\
	\textbf{Output} $n \times n$ matrices $(P^{*}, Z^{*})$
	\begin{algorithmic}[1]
		\State $\frac{n(n+1)}{2} = \text{rank}(Y_O)$
		\State $P^{*} = \argmin_{P} \mathcal{PF}(Y_O, Y_O^{+}, P, n) $
		\State \Return $P^{*}$
	\end{algorithmic}
\end{algorithm}

We now build some intuition towards the Algorithm~\ref{alg:learn_jls},~\ref{alg:switches}. A key observation (that we show below) is that the output observation matrix can be represented as
\begin{align}
Y_O &= \mathbb{E}_{\mu \in [s]^T}[C_{\mu} \otimes C_{\mu}]\mathbb{E}_{\sigma \in [s]^T}[B_{\sigma} \otimes B_{\sigma}] V \label{output_mat}\\
Y^{+}_O &= \mathbb{E}_{\mu \in [s]^T}[C_{\mu} \otimes C_{\mu}](\sum_{i=1}^s p_i A_i \otimes A_i)\mathbb{E}_{\sigma \in [s]^T}[B_{\sigma} \otimes B_{\sigma}] V  \nonumber
\end{align}
where $V = [\bar{v}_1 \otimes \bar{v}_1, \ldots, \bar{v}_d \otimes \bar{v}_d]$

To find the number of switches (or mode) we need to recover the matrix $\sum_{i=1}^s p_i A_i \otimes A_i$. This is $U^{\dagger} Y_O^{+} V^{\dagger} $ in Eq.~\eqref{rank_minim}. However, due to the non--uniqueness of matrix factorization we could have that
\begin{align*}
\hat{Y}^{+}_O &= U^{\dagger} Y_O^{+} V^{\dagger}  = Z^{-1}(\sum_{i=1}^s p_i A_i \otimes A_i) Z \\
(\sum_{i=1}^s p_i A_i \otimes A_i) &= \mathcal{L}^{-1}(\sum_{i=1}^s p_i \text{vec}(A_i)\text{vec}(A_i)^{'})
\end{align*}
for some unknown full rank matrix $Z$. Problem $\mathcal{PF}$ in Eq.~\eqref{rank_minim} is essentially a non--convex problem that needs to be solved. Intuitively, it means that the output observation matrix $\hat{Y}_{O}^{+}$ is similar to a positive semi--definite matrix (transformed by the swap transformation). The problem (in computation) arises because positive semi--definiteness, or the rank, is not preserved under swap transformation, $\mathcal{L}$. 

Further, the problem is not much easier even if we relax our constraints. Let us consider the following biconvex problem on the convex set
\[
\argmin_{P \in \mathbb{S}_n^{+}, \text{trace}(Z) \geq b > 0} ||\hat{Y}^{+}_O Z - Z \mathcal{L}^{-1}(P)||^2_{F}
\]
Usually, an alternating minimization based approach is used to solve such problems (See~\cite{gorski2007biconvex}). These problems typically have multiple local minima and uniqueness is not guaranteed. Perhaps the underlying difficulty lies in the fact that Eq.~\eqref{output_mat} resembles finding the components of a mixture model(See~\cite{kalai2012disentangling}). In fact, Assumption~\ref{independent_subspace} is remarkably close to Condition 1 in~\cite{anandkumar2013overcomplete}. They show that such a non--degeneracy condition is required for reconstruction of individual components (or modes) in latent variable or topic models.  

In general, without more information on switch sequences or number of discrete modes, we observe that finding the unique minimal number of modes is a hard problem.
\begin{remark}
\label{input_stochastic}
For our algorithms we assumed that the inputs, $v_j \in \mathbb{R}^{pT}$, were deterministic. However, this is not necessary for our analysis to work. Since we analyze in a expectation based setting, we only need to ensure that $\mathbb{E}[v_j v_j^{'}] =  \Sigma_j$, where $\text{vec}(\Sigma_j)$ spans the space $\mathbb{R}^{pT}$ for our results to hold.
\end{remark}

\section{Main Results}
\label{results}
Our main results hinge on the fact that we can find $T$ large enough so that maximum rank is achieved. 	
\begin{proposition}
\label{finite_time_reachability}
The matrix $C_T = \mathbb{E}_{\mu^T}[C_{\mu^T} \otimes C_{\mu^T}]$ achieves maximum rank at $T = n^2 + n - 1$, \textit{i.e.}, there is no increase in rank of $C_T$ for $T \geq n^2 + n$. A similar assertion holds for $B_T = \mathbb{E}_{\sigma^T}[B_{\sigma^T} \otimes B_{\sigma^T}]$
\end{proposition}
\begin{proof}
Clearly it is obvious that $\text{rank}(C_T) \leq \text{rank}(C_{T+1})$. Also, 
\begin{align*}
&\mathbb{E}_{\mu \in [s]^T}[C_{\mu} \otimes C_{\mu}] = C_T\\ 
=&\begin{bmatrix}C \otimes C \\
C \otimes C (\sum_{i=1}^s p_i A_i) \\
C \otimes C (\sum_{i=1}^s p_i A_i)^2 \\
\vdots\\
C \otimes C (\sum_{i=1}^s p_i A_i)^{T} \\
\vdots \\
C(\sum_{i=1}^s p_i A_i)^T \otimes C \\
(C(\sum_{i=1}^s p_i A_i)^{T-1} \otimes C)(\sum_{i=1}^s p_i A_i \otimes A_i)\\
\vdots \\
(C \otimes C) (\sum_{i=1}^s p_i A_i \otimes A_i)^T
\end{bmatrix}
\end{align*}
Further for any $T \geq n^2 + n-1$, we see that in $C_{T+1}$ the extra submatrices (compared to $C_T$) added are of one of the following form
\[
\Bigg\{(C(\sum_{i=1}^s p_i A_i )^{T+1-k} \otimes C)(\sum_{i=1}^s p_i A_i \otimes A_i)^{k}\Bigg\}_{k=0}^{T+1}
\]
or
\[
\Bigg\{(C \otimes C (\sum_{i=1}^s p_i A_i)^{T+1-k})(\sum_{i=1}^s p_i A_i \otimes A_i)^{k}\Bigg\}_{k=0}^{T+1}
\] 
Then for any submatrix in the first set, 
\[
C_T^1(k)= (C(\sum_{i=1}^s p_i A_i )^{T+1-k} \otimes C)(\sum_{i=1}^s p_i A_i \otimes A_i)^{k}
\]
For $k < n^2$, we have that $T+1-k \geq n$. By Cayley--Hamilton theorem we know that 
\[
(\sum_{i=1}^s p_i A_i )^{t} = \sum_{j=0}^{n-1}\lambda_j C(\sum_{i=1}^s p_i A_i )^j
\]
for all $t \geq n$. Then we can express $C_T^1(k)$ as a linear combination of $\{(C(\sum_{i=1}^s p_i A_i )^{j} \otimes C)(\sum_{i=1}^s p_i A_i \otimes A_i)^{k}\}_{j=0}^{n-1}$. Since $j + k \leq n^2 + n - 1$ due to the conditions on $j, k$, we have that the set of submatrices $\{(C(\sum_{i=1}^s p_i A_i )^{j} \otimes C)(\sum_{i=1}^s p_i A_i \otimes A_i)^{k}\}_{j=0}^{n-1}$ were already present in $C_{T}, C_{n^2 + n -1}$.

Next, for $k \geq n^2$, we use Cayley--Hamilton theorem on $\sum_{i=1}^s p_i A_i \otimes A_i$,
\[
(\sum_{i=1}^s p_i A_i  \otimes A_i)^{t} = \sum_{j=0}^{n^2-1}\lambda_j (\sum_{i=1}^s p_i A_i \otimes A_i )^j
\]
Then we can express $C_T^1(k)$ as a linear combination of $\{(C(\sum_{i=1}^s p_i A_i )^{T+1-k} \otimes C)(\sum_{i=1}^s p_i A_i \otimes A_i)^{j}\}_{j=0}^{n^2-1}$. Again since these $j+k \leq n^2 +n -1$, they were present in $C_T, C_{n^2 + n - 1}$ as well. This shows that the newly added rows as we increase from $C_T$ to $C_{T+1}$ only adds extra rows (submatrices) that are a linear combination of rows that are a part of $C_{n^2 + n - 1}$. As a result there is no increase in rank. Similarly, for the second set of matrices
\[
\Bigg\{(C \otimes C (\sum_{i=1}^s p_i A_i)^{T+1-k})(\sum_{i=1}^s p_i A_i \otimes A_i)^{k}\Bigg\}_{k=0}^{T+1}
\] 
Define $C^2_T(k)$ as
\[
C_T^2(k)= (C \otimes C(\sum_{i=1}^s p_i A_i )^{T+1-k})(\sum_{i=1}^s p_i A_i \otimes A_i)^{k}
\]
for $k < n^2$, then $T+1-k \geq n$ and we can express this as a linear combination of $\{(C \otimes C (\sum_{i=1}^s p_i A_i )^{j})(\sum_{i=1}^s p_i A_i \otimes A_i)^{k}\}_{j=0}^{n-1}$ (By using Cayley--Hamilton theorem on $\sum_{i=1}^s p_iA_i$). Observe that each of these submatrices were already a part of $C_{T}$ and $C_{n^2 + n -1}$. Next, for $k \geq n^2$ we can express $C_T^1(k)$ as a linear combination of $\{(C \otimes C(\sum_{i=1}^s p_i A_i )^{T-k})(\sum_{i=1}^s p_i A_i \otimes A_i)^{j}\}_{j=0}^{n^2-1}$ (By using Cayley--Hamilton theorem on $(\sum_{i=1}^s p_iA_i \otimes A_i)$). Observe that these submatrices were present in $C_T$ as well. Then we observe that for $T \geq n^2+n-1$ there is no increase in rank as we go from $C_T$ to $C_{T+1}$.
\end{proof}
\begin{lemma}
	\label{maximal_t}
If there exists some $T^{*}$ for which Assumption~\ref{independent_subspace} is true, then it is true for all $T > T^{*}$.
\end{lemma}
\begin{lemma}
	\label{full_rank}
	A direct consequence of Assumption~\ref{independent_subspace} is that $\mathbb{E}_{\mu \in [s]^{n^2+n-1}}[C_{\mu} \otimes C_{\mu}]$ is at least rank $\frac{n(n+1)}{2}$ if and only if it is weakly observable. A similar assertion holds for $\mathbb{E}_{\sigma \in [s]^{n^2+n-1}}[B_{\sigma} \otimes B_{\sigma}]$
\end{lemma}
\begin{proof}
If the JLS is weakly observable then there exists no $x \in \mathbb{R}^n$ such that for all $\mu \in [s]^T$
\[
C_{\mu} x = 0
\]
Thus $\mathbb{E}_{\mu \in [s]^T}[C_{\mu} x x^T C_{\mu}] \neq 0$. Next consider that indeed $\text{rank}(\mathbb{E}_{\mu \in [s]^T}[C_{\mu} \otimes C_{\mu}]) < \frac{n(n+1)}{2}$. This means that there exists a symmetric matrix $S$ such that 
\begin{align*}
&\mathbb{E}_{\mu \in [s]^T}[C_{\mu} \otimes C_{\mu}] \text{vec}(S) = 0 \\
&\mathbb{E}_{\mu \in [s]^T}[C_{\mu} S C^{'}_{\mu}] = 0 \\
&(\mathbb{E}_{\mu \in [s]^T}[C_{\mu} P_1 C^{'}_{\mu}])^{'}  = (\mathbb{E}_{\mu \in [s]^T}[C_{\mu} P_2 C^{'}_{\mu}])^{'}  
\end{align*}
Here $\langle P_1 , P_2 \rangle = 0$. However this contradicts Assumption~\ref{independent_subspace}. Thus $\text{rank}(\mathbb{E}_{\mu \in [s]^T}[C_{\mu} \otimes C_{\mu}]) \geq \frac{n(n+1)}{2}$. Converse follows trivially from Assumption~\ref{independent_subspace}. Let there exist $x$ such that  
\begin{equation*}
\mathbb{E}[C_{\mu} xx^{'}C_{\mu}^{'}] = 0
\end{equation*}
This contradicts Assumption~\ref{independent_subspace} by letting $P_1 = 0, P_2 = xx^{'}$.
\end{proof}
\begin{proposition}
	\label{rank_kron}
Under Assumptions~\ref{stability}--\ref{independent_subspace}, the matrix $\mathcal{H}_T$ denoted by
\begin{equation*}
\mathcal{H}_T = \mathbb{E}_{\mu, \sigma \in [s]^T}[\mathcal{H}_{\mu, \sigma} \otimes \mathcal{H}_{\mu, \sigma}]
\end{equation*}
has full rank for all $T \geq n^2 + n -1$.
\end{proposition}
\begin{proof}
From Definition~\ref{hankel}, it follows that 
\[
\mathcal{H}_T = \mathbb{E}_{\mu \in [s]^T}[C_{\mu} \otimes C_{\mu}]\mathbb{E}_{\sigma \in [s]^T}[B_{\sigma} \otimes B_{\sigma}]
\]
By the assumptions of weak controllability and observability combined with Proposition~\ref{finite_time_reachability}, Lemma~\ref{full_rank} we have that the rank of $\mathcal{H}_T=n^2$.
\end{proof}
\begin{theorem}
	\label{alg1}
	Let the output of Algorithm~\ref{alg:learn_jls} be $(Y_O, Y^{+}_O)$. Then, under Assumptions~\ref{stability}--\ref{independent_subspace}, we have that $\text{Rank}(Y_O) = \frac{n(n+1)}{2}$ for all $T \geq n^2 + n - 1$.
\end{theorem}
\begin{proof}
The proof follows from making the observation in Eq.~\eqref{output_mat}. Then we can use Lemma~\ref{full_rank}, Proposition~\ref{rank_kron} to conclude that $\text{Rank}(Y_O) = \frac{n(n+1)}{2}$ for all $T \geq n^2 + n - 1$ because we project on $V$ which is the subspace generated by vectors of the form $\text{vec}(S)$ where $S$ is symmetric.
\end{proof}
\begin{lemma}
	\label{rank_minim_lemma}
	The problem in Eq.~\eqref{rank_minim} is feasible for a JLS with inputs as described in Algorithm~\ref{alg:switches}.
\end{lemma}	
\begin{proof}
The proof of this follows easily from the discussion in preceding section.
\end{proof}
\begin{theorem}
	\label{alg2}
	Let the output of Algorithm~\ref{alg:switches} be $(P^{*}, Z^{*})$. Then, under Assumptions~\ref{stability}--\ref{independent_subspace} and if the JLS is a minimal realization, we have that $\text{Rank}(P^{*}) = s$ for all $T \geq n^2 + n - 1$.
\end{theorem}
\begin{proof}
Recall Eq.~\eqref{output_mat}. Then $SVD(Y_O) \rightarrow (U, \Sigma , V)$ then we have that $\hat{Y}_O^{+}$ in Step $4$ of Algorithm~\ref{alg:switches} can be written as 
\[
\hat{Y}_O^{+} = Z^{-1}(\sum_{i=1}^s p_i A_i \otimes A_i) Z
\]
for some $Z$. Since $s$ is minimal, we cannot have $s^{'} < s$ such that
\[
Z \hat{Y}_O^{+} Z^{-1} = (\sum_{i=1}^{s^{'}} \hat{p}_i \hat{A}_i \otimes \hat{A}_i) 
\]
Then let $P = \sum_{i=1}^s p_i\text{vec}(A_i)\text{vec}(A_i)^{'}$ is indeed positive definite with rank $s$. From Proposition~\ref{switches} we have that 
\[
\mathcal{L}(Z \hat{Y}_O^{+} Z^{-1}) = P
\]
Now setting $P=P^{*}, Z = Z^{*}$, we have our assertion.
\end{proof}
In this section we show the correctness of our algorithms. First, in Proposition~\ref{finite_time_reachability}, we show that there exists a maximal value $T$ used in Algorithm~\ref{alg:learn_jls} such that the rank of $\mathbb{E}_{\mu^t, \sigma^t}[\mathcal{H}_{\mu^t, \sigma^t} \otimes \mathcal{H}_{\mu^t, \sigma^t}]$ does not increase for any $t \geq T$. Although $T$ is itself a function of the state space dimension, it is still important to know a bound on it. For example, suppose we knew an upper bound, $a$, on the state space dimension then we could recover the true state space dimension by just setting $T = a^2 + a -1$. We show in Theorem~\ref{alg1} that Algorithm~\ref{alg1} indeed gives us the true state space dimension. Based on our knowledge of $T$ we can now find the minimal number of discrete modes. Theorem~\ref{alg2} states that, for a minimal JLS, Algorithm~\ref{alg:switches} gives us a positive semi--definite matrix that has the same rank as the minimal number of discrete modes. However, finding this number is hard because Eq.~\eqref{rank_minim} is a non--convex problem.

\section{Conclusion}
\label{conclusions}
In this paper we studied minimal realization problems in a JLS when only continuous input and output are observed. We approached the problem by first finding the span of hidden state space. We showed that there exists a maximum value of $T$, or information complexity, for which a certain observation matrix achieves full rank. This rank was then used to find the number of discrete modes of the JLS. We showed that this was equivalent to solving a non--convex problem with possibly multiple solutions. Despite this, there exists a notion of minimality which characterizes the minimal pair of state dimension and discrete modes to completely span the observed output space. We also discussed how computing this minimal pair is hard as it has possible connections to component identification in mixture models, topic models etc.  

There are several important directions of future work. Foremost is that we do not address issue the question of sample complexity. Specifically, we do not know if the algorithmic approach we discuss here is efficient in terms of number of JLS copies required. Finding a sample efficient algorithm to recover the minimal parameters seems a reasonable next course of action. 

Second we still do not have a complete algorithm to learn $p_i$s. However, there is substantial work on learning Gaussian mixtures(See~\cite{ge2015learning}). One possible solution to this might be to transform the output observation into a covariance matrix that is a mixture of Gaussians(which correspond to the discrete modes). We can then recover all parameters, \textit{i.e.}, $\{p_i, A_i\}_{i=1}^s$ using similar techniques as Huang et. al.

The work here is complementary to the work on balanced truncation in JLS in~\cite{kotsalis2008balanced} where the number of discrete modes are known and fixed apriori. It will be interesting to understand the error in model misspecification if we choose $T < n^2 + n - 1$ and how the truncated state dimension affects the minimal number of nodes.
\bibliographystyle{IEEEtran}
\bibliography{biblio.bib}
\end{document}